\documentclass[12pt,leqno,letterpaper]{article}

\usepackage{amsmath,amsthm,enumerate,amssymb}
\usepackage[utf8]{inputenc}
\usepackage[english]{babel}

\newtheorem{theorem}{Theorem}[section]
\newtheorem{proposition}[theorem]{Proposition}
\newtheorem{lemma}[theorem]{Lemma}
\newtheorem{corollary}[theorem]{Corollary}
\newtheorem{definition}[theorem]{Definition}

\newcommand{\tr}{{\rm Tr\hskip -0.2em}~}
\DeclareMathOperator{\frechetdiff}{\mathit d}
\newcommand{\fd}[1]{\hskip-0.2em\frechetdiff\hskip -0.23em{#1}}
\newcommand{\closefd}[1]{\hskip-0.2em\frechetdiff\hskip -0.35em{#1}}

\begin{document}

\title{Quantum entropy derived from first principles}
\author{Frank Hansen}
\date{}

\maketitle

\begin{abstract} The most fundamental properties of quantum entropy are derived by considering the union of two ensembles. We discuss the limits these properties put on an entropy measure and obtain that they uniquely determine the form of the entropy functional up to normalisation. In particular, the result implies that all other properties of quantum entropy may be derived from these first principles.\\[1ex]
\end{abstract}

\section{Introduction}

Von Neumann suggested in 1927 the function
\[
S(\rho)=-\tr\rho\log\rho
\]
as a measure of quantum entropy, where $ \rho $ is the state of the ensemble under consideration. The form of the entropy measure was derived from a gedanken experiment in phenomenological thermodynamics. One may define a more general ``entropy measure'' by setting
\[
S_f(\rho)=-\tr f(\rho)
\]
for an arbitrary convex function $ f\colon(0,\infty)\to\mathbf R. $ It is known  \cite[Section E]{kn:wehrl:1978}  that additivity 
\[
S_f(\rho_1\otimes\rho_2)=S_f(\rho_1)+S_f(\rho_2)
\]
then implies that $ f $ is of the form $ f(t)=t\log t $ up to a positive multiplicative constant. We will discuss characterisations of entropy that are more related to physical principles.
The von Neumann entropy enjoys two basic properties:

\begin{enumerate}[(i)]

\item\label{first principle 1} The entropy of the union of two ensembles is greater than or equal to the average entropy of the component ensembles.

\item\label{first principle 2} The incremental information increases when two ensembles are united.

\end{enumerate}

The first principle is satisfied by the requirement that $ f $ is convex. Indeed,  the ``entropy measure'' defined above then has the property that the map
\begin{equation}\label{first rule}
\rho\to S_f(\rho)
\end{equation}
is concave; and this is the mathematical expression of the first principle.
The second principle is interpreted as convexity of the map
\begin{equation}\label{second rule}
\rho\to S_f(\rho_1)-S_f(\rho)
\end{equation}
in positive definite operators on a bipartite system $ \mathcal H=\mathcal H_1\otimes\mathcal H_2 $ where $ \rho_1 $ denotes the partial trace of $ \rho $ on $ \mathcal H_1. $ Lieb and Ruskai obtained that the von Neumann entropy enjoys this property \cite[Theorem 1]{kn:lieb:1973:3}, cf. also \cite{kn:lieb:1973:4} for a broader discussion and \cite{kn:hansen:2016} for a truly elementary proof.

The main result of this paper is that the von Neumann entropy is uniquely determined up to normalisation by the requirements in (\ref{first rule}) and (\ref{second rule}) representing the first principles (\ref{first principle 1}) and (\ref{first principle 2}).  This clarifies a long-standing problem in quantum physics. If there were other substantial different ways of defining quantum entropy, then it could happen that some properties derived for a specific physical system were mere mathematical artefacts of the chosen entropy function. They would be in accordance with the underlying physical principles stated above, but they would also reflect an arbitrary mathematical choice. We now know that this cannot happen.

The von Neumann entropy may increase when passing to a subsystem, cf. the remarks in (d) of \cite{kn:lieb:1973:4}. This is called the intuitive defect in quantum physics. We now realise that this defect cannot be remedied by possibly adopting an alternative definition of quantum entropy.

\section{Entropic functions}

\begin{definition}
Let $ f\colon(0,\infty)\to\mathbf R $ be a convex function. We say that
$ f $ is entropic if the function
\begin{equation}\label{definition of an entropic function}
F(\rho)=-\tr f(\rho_1)+\tr f(\rho)
\end{equation}
is convex in positive definite operators $ \rho $ on any finite dimensional bipartite system 
$ \mathcal H=\mathcal H_1\otimes\mathcal H_2 $ where $ \rho_1 $ denotes the partial trace of $ \rho $ on $ \mathcal H_1. $
\end{definition}

The function $ f(t)=t\log t $ is thus entropic.
We recall that a quantum channel is represented by a completely positive trace preserving map 
$ \Phi\colon B(\mathcal  H)\to B(\mathcal K) $ between Hilbert spaces $ \mathcal H $ and $ \mathcal K. $ 

\begin{lemma}\label{lemma: entropic and strongly entropic are the same}
Asume $ f\colon(0,\infty)\to\mathbf R $ is entropic. The entropy gain
\begin{equation}\label{second rule +}
\rho\to -\tr_{\mathcal K} f\bigl(\Phi(\rho)\bigr)+\tr_{\mathcal H} f(\rho)
\end{equation}
over a quantum channel $ \Phi $ is then convex.
\end{lemma}

\begin{proof}
By Stinespring's theorem there exists a finite dimensional Hilbert space $ \mathcal R $ and a linear map $ W\colon\mathcal H\to \mathcal K\otimes \mathcal R $ such that $ W^*W=I_{\mathcal H} $ and
\[
\Phi(x)=\tr_{\mathcal R} WxW^*
\]
for every $ x\in B(\mathcal H). $ We may assume $ f(0)=0 $ and then trivially
\[
\tr_{\mathcal K\otimes\mathcal R} f(W\rho W^*)=\tr_{\mathcal H}f(\rho).
\]
It now follows by (\ref{definition of an entropic function}) that
\[
\rho\to -\tr_{\mathcal K} f\bigl(\Phi(\rho)\bigr)+\tr_{\mathcal H} f(\rho)=-\tr_{\mathcal K} f(\tr_{\mathcal R} W\rho W^*)+\tr_{\mathcal K\otimes\mathcal R} f(W\rho W^*)
\]
is convex.
\end{proof}

\begin{proposition}
Let $ f\colon(0,\infty)\to\mathbf R $ be an entropic function. The operator function of $ k $ variables
\begin{equation}\label{function of k variables associated with an entropic function}
G(\rho_1,\dots,\rho_k)=-\tr f(\rho_1+\cdots+\rho_k) + \tr f(\rho_1)+\cdots+\tr f(\rho_k)
\end{equation}
is for any natural number $ k $ convex in positive definite operators $ \rho_1,\dots,\rho_k $ on any finite dimensional Hilbert space.
\end{proposition}

\begin{proof}
We may assume $ k\ge 2 $ and consider the bipartite system
\[
\mathcal H=\mathcal H_1\otimes l^2(0,1,\dots,k-1)=\mathcal H_1\oplus\mathcal H_1\oplus\cdots\oplus\mathcal H_1
\]
where the partial trace is given by
\[
\begin{pmatrix}
a_{11} & a_{12} & \dots & a_{1k}\\
a_{21} & a_{22} & \dots & a_{2k}\\
\vdots & \vdots  &          &  \vdots\\
a_{k1} & a_{k2} & \dots & a_{kk}
\end{pmatrix}_1
=a_{11}+a_{22}+\cdots+a_{kk}.
\]
We now apply the assumption of convexity of the function
\[
A\to -\tr_1 f(A_1)+\tr_{12} f(A)
\]
in the convex set of positive definite diagonal block matrices
\[
A=\begin{pmatrix}
               \rho_1 & 0          & \dots  & 0\\
               0         & \rho_2  &           & 0\\
               \vdots &              & \ddots & \\
               0         & 0           &           & \rho_n           
                                                   \end{pmatrix}
\]
on $ \mathcal H, $ and obtain that the operator function
\[
\begin{array}{l}
G(\rho_1,\dots,\rho_k)=G(A)= -\tr_1 f(A_1)+\tr_{12} f(A)\\[2ex]
=
-\tr_1 f(\rho_1+\cdots+\rho_k) + \tr_1 f(\rho_1)+\cdots+\tr_1 f(\rho_k)
\end{array}
\]
is convex.
\end{proof}

\subsection{Subentropic functions}

\begin{definition}
Let $ f\colon(0,\infty)\to\mathbf R $ be a convex function. We say that
$ f $ is subentropic of order $ k $ if the function
\begin{equation}\label{function of k variables associated with a subentropic function}
G(\rho_1,\dots,\rho_k)=-\tr f(\rho_1+\cdots+\rho_k) + \tr f(\rho_1)+\cdots+\tr f(\rho_k)
\end{equation}
is convex in positive definite operators $ \rho_1,\dots,\rho_k $ on any finite dimensional Hilbert space. We say that $ f $ is subentropic if it is subentropic of all orders.
\end{definition}

We notice that an entropic function is subentropic. 

\section{Analysis of subentropic functions}

\subsection{The Fréchet differential}

Let $ f\colon (0,\infty)\to\mathbf R $ be a continuously differentiable function. The action of the  Fréchet differential $ \closefd{}f(\rho) $ in $ h, $ where $ \rho $ is positive definite and $ h $ is self-adjoint, may be defined by setting
\[
\closefd{}f(\rho)h=\lim_{\varepsilon\to 0}\frac{f(\rho+\varepsilon h)-f(\rho)}{\varepsilon}\,.
\]
Notice that $ \rho+\varepsilon h $ eventually is positive definite. The construction therefore depends on spectral theory. We shall only work with finite dimensional Hilbert spaces  in which case the Fréchet differential may be expressed as the Hadamard product
\[
\closefd{}f(\rho)h=L_f(\rho)\circ h
\]
of $ h $ and the L{\"o}wner matrix $ L_f(\rho) $ in a basis that diagonalises $ \rho. $ This readily extends the action of the Fréchet differential to operators that are not necessarily self-adjoint.
Consider now the bivariate function
\[
k(t,s)=\frac{f(t)-f(s)}{t-s}=\int_0^1 f'(\lambda t+(1-\lambda)s)\,d\lambda\qquad t,s>0.
\]
If $L_\rho$ and $R_\rho$ denote left and right multiplication with operators $\rho$ on a Hilbert space $ \mathcal H $ of finite dimension $ n, $ then
\[
\tr h^*\closefd{}f(\rho) h=\sum_{i,j=1}^n |(he_i\mid e_j)|^2 \frac{f(\lambda_i)-f(\lambda_j)}{\lambda_i-\lambda_j}=\tr h^*\, k(L_\rho,R_\rho) h,
\]
where the intermediary calculation is carried out in an orthonormal basis $ (e_1,\dots,e_n) $ of eigenvectors of $ \rho $ with corresponding eigenvalues $ \lambda_1,\dots,\lambda_n $ counted with multiplicity. We may thus identify the Fréchet differential $ \closefd{}f(\rho) $ with $ k(L_\rho,R_\rho). $ If $ f $ has strictly positive derivative, then $ k $ is positive. The Fréchet differential is then a positive definite operator with inverse
\[
\closefd{}f(\rho)^{-1}= k(L_\rho,R_\rho)^{-1}
\]
for positive definite $ \rho. $ We note that $ \closefd{}f(\rho) $ acts as multiplication 
$
\closefd{}f(t)h=f'(t)h
$
with $ f'(t), $ when $ \rho=t $ is a multiple of the identity.

\subsection{Subentropic functions are smooth}

\begin{theorem}
A function subentropic of order two is operator convex, the derivative is operator monotone, and the second derivative is convex.
\end{theorem}

\begin{proof}
Assume $ f\colon(0,\infty)\to\mathbf R $ is subentropic of order two. The function
\[
G(\rho,\sigma)=-\tr f(\rho+\sigma)+\tr f(\rho)+\tr f(\sigma)
\]
is then by definition convex in positive definite $ \rho $ and $ \sigma. $ By fixing $ \sigma>0 $ we obtain that the function of one variable
\begin{equation}\label{convex function of one variable}
G(\rho)=-\tr f(\rho+\sigma)+\tr f(\rho)
\end{equation}
is convex.   Let $ \varphi $ be a positive and even $ C^\infty $-function defined in the real line, vanishing outside the closed interval $ [-1,1] $ and normalised such that 
\[
\int_{-1}^1 \varphi(t)\,dt=1.
\]
For $ \varepsilon>0 $ we consider the regularisation
\[
f_\varepsilon(t)=\int_{-1}^1 \varphi(s) f(t-\varepsilon s)\,ds\qquad t>\varepsilon
\]
which is a convex function. By applying the convexity in (\ref{convex function of one variable}) we furthermore obtain that the function
\[
\begin{array}{rl}
G_\varepsilon(\rho)&=-\tr f_\epsilon(\rho+\sigma)+\tr f_\epsilon(\rho)\\[2ex]
&=\displaystyle\int_{-1}^1 \varphi(s)\tr\bigl(-f(\rho+\sigma-\varepsilon s)+f(\rho-\varepsilon s)\bigr)\,ds
\end{array}
\]
is convex in positive definite $ \rho>\varepsilon. $ By the equivalent expression
\[
f_\varepsilon(t)=\frac{1}{\varepsilon}\int_0^\infty\varphi\left(\frac{t-s}{\varepsilon}\right)f(s)\,ds
\qquad t>\varepsilon,
\]
we realise that $ f_\varepsilon $ is infinitely differentiable. By replacing $ \sigma $ with $ t\sigma $ for a real $ t>0 $ we realise that the function
\[
\rho\to\frac{-\tr f_\epsilon(\rho+t\sigma)+\tr f_\epsilon(\rho)}{t}
\]
is convex, and since $ f_\epsilon $ is continuously differentiable we obtain by letting $ t $ tend to zero that the Fréchet derivative
\[
\rho\to-\tr \closefd{f_\epsilon}(\rho)\sigma=-\tr f'_\varepsilon(\rho)\sigma 
\]
is convex in positive definite $ \rho>\varepsilon $ for all positive definite $ \sigma. $ It follows that $ f'_\varepsilon $ is operator concave in the interval $ (\varepsilon,\infty). $ Since $ f'_\varepsilon $ is non-decreasing we furthermore obtain that $ f'_\varepsilon $ is operator monotone. The function $ t\to f'_\varepsilon(t+\varepsilon) $ is thus operator monotone in the positive half-line and may therefore be written on the canonical form
\[
f'_\varepsilon(t+\varepsilon)=\alpha+\beta t+\int_0^\infty \left(\frac{\lambda}{1+\lambda^2}-\frac{1}{t+\lambda}\right)\fd{}\nu_\varepsilon(\lambda)
\qquad t>0
\]
for some non-negative measure $ \nu_\varepsilon $ with
\[
\int_\varepsilon^\infty (1+\lambda^2)^{-1} \fd{}\nu_\varepsilon(\lambda)<\infty 
\]
and constants $ \alpha,\beta $ (depending on $ \varepsilon) $ with $ \beta\ge 0, $ cf. \cite[Theorem 5.2]{kn:hansen:2013:1}. From this formula it readily follows that $ f_\varepsilon $ is operator convex. Since by convexity $ f $ is continuous, it is well-known that $ f_\varepsilon $ eventually converges uniformly towards $ f $ on any compact subset of $ (0,\infty) $ as $ \varepsilon $ tends to zero. Therefore, $ f $ is operator convex (indeed, point-wise convergence would suffice). As a consequence, $ f $ is infinitely differentiable and $ f'_\epsilon $ eventually converges towards $ f' $ as $ \varepsilon $ tends to zero. We conclude that $ f' $ is operator monotone and that $ f'' $ is convex.
\end{proof}

We proved that a subentropic function is operator convex; therefore it is also real analytic. Furthermore, its second derivative is either constantly zero or strictly positive.

\begin{theorem}\label{subentropicity of order two implies subentropicity}
A function  $ f\colon (0,\infty)\to\mathbf R $ is subentropic if and only if it is subentropic of order two.
\end{theorem}

\begin{proof} Suppose that $ f $ is subentropic of order two.
Since $ f $ is (even infinitely) differentiable we may apply the chain rule to the function $ G(\rho_1,\dots,\rho_k) $ defined in (\ref{function of k variables associated with a subentropic function}) and calculate the first Fréchet differential
\[
\begin{array}{l}
\fd{}G(\rho_1,\dots,\rho_k)h=\fd{}_1G(\rho_1,\dots,\rho_k)h_1+\cdots+\fd{}_kG(\rho_1,\dots,\rho_k)h_k\\[2ex]
=-\tr_1\closefd{}f(\rho_1+\cdots+\rho_k)h_1+\tr_1 \closefd{}f(\rho_1)h_1\\[1ex]
\hskip 5em +\cdots\\[1ex]
\hskip 5em -\tr_1\closefd{}f(\rho_1+\cdots+\rho_k)h_k+\tr_1 \closefd{}f(\rho_k)h_k\\[2ex]
=\tr_1\Bigl[-\closefd{}f(\rho_1+\cdots+\rho_k)(h_1+\cdots+h_k)+\closefd{}f(\rho_1)h_1+\cdots+\closefd{}f(\rho_k)h_k\Bigr]\\[2ex]
=\tr_1\Bigl[-f'(\rho_1+\cdots+\rho_k)(h_1+\cdots+h_k)+f'(\rho_1)h_1+\cdots+f'(\rho_k)h_k\Bigr]
\end{array}
\]
in vectors $ h=(h_1,\dots,h_k) $ of self-adjoint matrices.
We continue to calculate the second Fréchet differential
\[
\begin{array}{l}
d^2 G(\rho_1,\dots,\rho_k)(h,h)=\fd{}\bigl(\fd{}G(\rho_1,\dots,\rho_k)h\bigr)h\\[1ex]
=\displaystyle\sum_{i=1}^k \fd{}_i\bigl(\fd{}G(\rho_1,\dots,\rho_k)h\bigr)h_i\\[1ex]
=\displaystyle\sum_{i=1}^k \fd{}_i\bigl(\tr_1\Bigl[-f'(\rho_1+\cdots+\rho_k)(h_1+\cdots+h_k)\Bigr]  \bigr)h_i\\[-1ex]
\hskip 5em+\displaystyle\sum_{i=1}^k \fd{}_i\bigl(\tr_1\Bigl[f'(\rho_1)h_1+\cdots+f'(\rho_k)h_k\Bigr]  \bigr)h_i\\[3ex]
=-\displaystyle\sum_{i=1}^k \tr_1 (h_1+\cdots+h_k)\closefd{}f'(\rho_1+\cdots+\rho_k)h_i+\sum_{i=1}^k \tr_1 h_i\closefd{}f'(\rho_i)h_i\\[4ex]
=-\tr_1(h_1+\cdots+h_k)\closefd{}f'(\rho_1+\cdots+\rho_k)(h_1+\cdots+h_k)\\
\hskip 5em+\displaystyle\sum_{i=1}^k \tr_1 h_i\closefd{}f'(\rho_i)h_i
\end{array}
\]
and thus obtain that the function $ G $ defined in  (\ref{function of k variables associated with a subentropic function}) is convex if and only if
\begin{equation}\label{entropy: full condition for convexity, k variables}
\begin{array}{l}
\tr_1(h_1+\cdots+h_k)\closefd{}f'(\rho_1+\cdots+\rho_k)(h_1+\cdots+h_k)\\[2ex]
\le \tr_1\Bigl[h_1\closefd{}f'(\rho_1)h_1+\cdots+h_k\closefd{}f'(\rho_k)h_k\Bigr]
\end{array}
\end{equation}
for positive definite $ \rho_1,\dots,\rho_k $ and self-adjoint $ h_1,\dots,h_k. $ However, since $ G $ by assumption is convex for $ k=2 $ we derive that
\begin{equation}\label{entropy: condition for convexity, two variables}
\tr_1(h_1+h_2)\closefd{}f'(\rho_1+\rho_2)(h_1+h_2)
\le \tr_1\Bigl[h_1\closefd{}f'(\rho_1)h_1+h_2\closefd{}f'(\rho_2)h_2\Bigr]
\end{equation}
for positive definite $ \rho_1,\rho_2 $ and self-adjoint $ h_1,h_2. $ By properly setting parantheses we now derive (\ref{entropy: full condition for convexity, k variables}) by repeated application of (\ref{entropy: condition for convexity, two variables}). This shows that $ f $ is subentropic of all orders and therefore subentropic.
\end{proof}

\begin{theorem}\label{theorem: equivalent condition for subentropicity}
A twice continuously differentiable function $ f\colon(0,\infty)\to\mathbf R $ with strictly positive second derivative is subentropic if and only if
\begin{equation}\label{condition for subentropicity in terms of Frechet differentials}
\closefd{}f'(\rho+\sigma)^{-1}\ge\closefd{}f'(\rho)^{-1}+\closefd{}f'(\sigma)^{-1}
\end{equation}
for positive definite operators $ \rho $ and $ \sigma $ on any finite dimensional Hilbert space.
\end{theorem}

\begin{proof} Convexity of the function
\[
G(\rho,\sigma)=-\tr f(\rho+\sigma)+\tr f(\rho)+\tr f(\sigma)
\]
is by (\ref{entropy: condition for convexity, two variables}) equivalent to the inequality
\[
\tr (a+b)\closefd{}f'(\rho+\sigma)(a+b)\le \tr a\closefd{}f'(\rho)a+ \tr b\closefd{}f'(\sigma)b
\]
for positive definite $ \rho,\sigma $ and self-adjoint $ a,b. $ Since the second Fréchet differential of $ G $
is a symmetric bilinear form this is again equivalent to the inequality
\begin{equation}\label{entropy: condition for convexity}
\tr (a+b)^*\closefd{}f'(\rho+\sigma)(a+b)\le \tr a^*\closefd{}f'(\rho)a+ \tr b^*\closefd{}f'(\sigma)b
\end{equation}
for positive definite $ \rho,\sigma $ and arbitrary $ a,b. $ 
Ando \cite[Remark on page 208]{kn:ando:1979} noticed without proof that the harmonic mean $ H_2(A,B) $ of two positive definite matrices $ A $ and $ B $ may be characterised as the maximum of all Hermitian operators $ C $ such that
\[
\begin{pmatrix}
C & C\\
C & C
\end{pmatrix}
\le 2\begin{pmatrix}
        A & 0\\
        0 & B
        \end{pmatrix}.
\]
For a proof we refer the reader to \cite[Footnote 4]{kn:hansen:2006:3}.
The condition in (\ref{entropy: condition for convexity}) may equivalently be written as
\[
\begin{array}{l}
\left(\begin{pmatrix}
           a\\
           b
           \end{pmatrix}
\left|
\begin{pmatrix}
\closefd{}f'(\rho+\sigma) & \closefd{}f'(\rho+\sigma)\\
\closefd{}f'(\rho+\sigma) & \closefd{}f'(\rho+\sigma)
\end{pmatrix}
\begin{pmatrix}
           a\\
           b
           \end{pmatrix}
\right.\right)_\tr\\[4ex]
\hskip 7em\le
\left(\begin{pmatrix}
           a\\
           b
           \end{pmatrix}
\left|
\begin{pmatrix}
\closefd{}f'(\rho) & 0\\
0               & \closefd{}f'(\sigma)
\end{pmatrix}
\begin{pmatrix}
           a\\
           b
           \end{pmatrix}
\right.\right)_\tr
\end{array}
\]
for positive definite $ \rho,\sigma $ and arbitrary $ a,b. $ The condition in (\ref{entropy: condition for convexity}) is thus equivalent to the inequality
\begin{equation}\label{entropy: second condition for convexity}
\closefd{}f'(\rho+\sigma)\le H_2\Bigl(\frac{1}{2}\closefd{}f'(\rho),\frac{1}{2}\closefd{}f'(\sigma)\Bigr)=\frac{1}{2}H_2\bigl(\closefd{}f'(\rho),\closefd{}f'(\sigma)\bigr)
\end{equation}
for positive definite $ \rho $ and $ \sigma; $ where we used that the harmonic mean is positively homogeneous. Since the inverse of the harmonic mean satisfies
\[
H_2(A,B)^{-1}=\frac{A^{-1}+B^{-1}}{2}
\]
we obtain by taking the inverses of both sides of inequality (\ref{entropy: second condition for convexity}) the equivalent inequality
\[
\closefd{}f'(\rho+\sigma)^{-1}\ge\closefd{}f'(\rho)^{-1}+\closefd{}f'(\sigma)^{-1}
\]
for positive definite $ \rho,\sigma. $ We have proved that $ f $ is subentropic of order two if and only if (\ref{condition for subentropicity in terms of Frechet differentials}) holds.  The assertion now follows from Theorem~\ref{subentropicity of order two implies subentropicity}.
\end{proof}

\begin{corollary}\label{corollary: super-additive function}
Let $ f\colon(0,\infty)\to\mathbf R $ be a non-affine subsentropic function. The positive and infinitely differentiable function
\[
g(t)=\frac{1}{f''(t)}\qquad t>0
\]
is super-additive in the sense that
\begin{equation}\label{condition for super-additivity}
g(t+s)\ge g(t)+g(s)\qquad\text{for}\quad t,s>0.
\end{equation}
It is increasing and may be extended to a continuous function with $ g(0)=0, $ meaning that $ f''(t)\to\infty $ as $ t $ tends to zero. 
\end{corollary}

\begin{proof}
Since $ f $ is non-affine the second derivative $ f'' $ is strictly positive, so $ g $ is well-defined. The super-additivity follows from Theorem~\ref{theorem: equivalent condition for subentropicity}.
Since $ g $ is positive the super-additivity implies that $ g $  is increasing. The limit
\[
g(0)=\lim_{s\to 0} g(s)=0
\]
follows by letting $ s $ tend to zero in (\ref{condition for super-additivity}). 
\end{proof}

\begin{theorem}
The convex function $ f(t)=-\log t, $ defined in the positive half-line, is subentropic; meaning that the multivariate operator function
\[
(\rho_1,\dots,\rho_k)\to\tr\log(\rho_1+\cdots+\rho_k)-\tr\bigl(\log\rho_1+\cdots+\log\rho_k\bigr)
\]
is convex in positive definite operators on a finite dimensional Hilbert space.
\end{theorem}

\begin{proof} 
The derivative $ f'(t)=-t^{-1} $ has strictly positive derivative. The Fréchet differential $ \closefd{}f'(\rho) $ is therefore a strictly positive operator on $ B(\mathcal H), $ where $ \mathcal H $ is the underlying Hilbert space. Since $ \closefd{}f'(\rho)h=\rho^{-1}h\rho^{-1} $ we obtain
\[
\closefd{}f'(\rho)^{-1}h=\rho h \rho.
\]
The function $ f(t)=-\log t $ is thus subentropic, by Theorem~\ref{theorem: equivalent condition for subentropicity}, if
\[
\tr h^*(\rho+\sigma)h(\rho+\sigma)\ge\tr h^*\rho h\rho+\tr h^*\sigma h\sigma
\]
for positive definite $ \rho,\sigma $ and arbitrary $ h. $ But this inequality reduces to
\[
\tr h^*\rho h\sigma+\tr h^*\sigma h\rho\ge 0
\]
which is trivially satisfied.
\end{proof}

\section{The main result}

Matrix entropies were introduced by Chen and Tropp  as a tool to obtain concentration inequalities for random matrices \cite{kn:tropp:2014}, and their representing functions may be characterised in various ways \cite{kn:hansen:2014:3}. Prominent examples are the functions $ f(t)=t\log t $ and $ f(t)=t^p $ for $ 1\le p\le 2. $ 

\begin{theorem}
An entropic function is a matrix entropy.
\end{theorem}

\begin{proof} Let $ \mathcal H $ be a finite dimensional Hilbert space and assume that $ f $ is entropic. By Lemma~\ref{lemma: entropic and strongly entropic are the same} we know that the entropic gain
\[
F(\rho)= -\tr f(\Phi(\rho)) + \tr f(\rho)
\]
over any quantum channel $ \Phi\colon B(\mathcal  H)\to B(\mathcal K) $ is a convex function in positive definite $ \rho. $ 
Since $ f $ is twice continuously differentiable, $ F $ is twice Fréchet differentiable on finite dimensional spaces, and the first Fréchet differential is given by
\[
\begin{array}{rl}
\closefd{}F(\rho)h&=-\tr_{\mathcal K}\closefd{}f\bigl(\Phi(\rho)\bigr)\Phi(h) + \tr_{\mathcal H}\closefd{}f(\rho)h\\[1ex]
&=-\tr_{\mathcal K} f'\bigl(\Phi(\rho)\bigr)\Phi(h) + \tr_{\mathcal H} f'(\rho)h.
\end{array}
\]
The second Fréchet differential is then calculated to be
\[
\fd{}^2F(\rho)(h,h)=-\tr_{\mathcal K} \Phi(h)\closefd{}f'\bigl(\Phi(\rho)\bigr)\Phi(h) + \tr_{\mathcal H} h\closefd{}f'(\rho)h
\]
in positive definite $ \rho $ and self-adjoint $ h. $ The convexity condition for $ F $ is therefore equivalent to the inequality
\[
\tr_{\mathcal K} \Phi(h)^*\closefd{}f'\bigl(\Phi(\rho)\bigr)\Phi(h) 
\le
\tr_{\mathcal H} h^*\closefd{}f'(\rho)h
=
\tr_{\mathcal K} \Phi(h^*\closefd{}f'(\rho)h)
\]
for positive definite $ \rho $ and arbitrary $ h, $ where we again used that the second Fréchet differential is a symmetric bilinear form. Consider the block matrices
\[
U=\frac{\sqrt{2}}{2}\begin{pmatrix}
     1 & -1\\
     0 & 0
     \end{pmatrix}
\qquad\text{and}\qquad
V=\frac{\sqrt{2}}{2}\begin{pmatrix}
0 & 0\\
1 & 1
\end{pmatrix}
\]
defined on the direct sum $ \mathcal H\oplus \mathcal H $ and put
\[
\Phi(X)=UXU^*+VXV^*
\]
for $ X\in B(\mathcal H\oplus \mathcal H). $ Then $ \Phi $ is completely positive and satisfies
\[
\Phi\begin{pmatrix}
                 \rho & a\\
                 b & \sigma
                 \end{pmatrix}=
\begin{pmatrix}
\displaystyle\frac{\rho+\sigma}{2}-\frac{a+b}{2} & 0\\
0 & \displaystyle\frac{\rho+\sigma}{2}+\frac{a+b}{2} 
\end{pmatrix}.
\]
We notice that $ \Phi $ is trace preserving. In particular, for diagonal block matrices
\[
\Phi\begin{pmatrix}
                 \rho & 0\\
                 0     & \sigma
                 \end{pmatrix}
=\frac{\rho+\sigma}{2}\begin{pmatrix}
                 1 & 0\\
                 0 & 1
                 \end{pmatrix},
\]
so $ \Phi $ is also unital. Setting
\[
h=\begin{pmatrix}
               a & 0\\
               0 & b
               \end{pmatrix}
\qquad\text{and}\qquad
A=\begin{pmatrix}
               \rho & 0\\
               0 & \sigma
               \end{pmatrix}
\]
we readily obtain from the definition of the Fréchet differential that
\[
h^*\closefd{}f'(A)h=\begin{pmatrix}
                       a^*\closefd{}f'(\rho)a & 0\\
                       0                            & b^*\closefd{}f'(\sigma)b
                       \end{pmatrix}
\]
and thus
\[
\Phi\bigl(h^*\closefd{}f'(A)h\bigr)=\frac{a^*\closefd{}f'(\rho)a +b^*\closefd{}f'(\sigma)b}{2}\begin{pmatrix}
                 1 & 0\\
                 0 & 1
                 \end{pmatrix}.
\]
On the other hand
\[
\Phi(h)^*\closefd{}f'\bigl(\Phi(A)\bigr)\Phi(h)
=\displaystyle\Big[\Bigl(\frac{a+b}{2}\Bigr)^*\closefd{}f'\Bigl(\frac{\rho+\sigma}{2}\Bigr)\Bigl(\frac{a+b}{2}\Bigr)\Bigr]
\begin{pmatrix}
                 1 & 0\\
                 0 & 1
                 \end{pmatrix}.
\]
By convexity of $ F $ and by taking the trace and dividing by $ 2, $ we thus obtain
\[
\tr\Bigl(\frac{ a+b}{2}\Bigr)^*\closefd{}f'\Bigl(\frac{\rho+\sigma}{2}\Bigr)\Bigl(\frac{a+b}{2}\Bigr)
\le\frac{1}{2}\tr a^*\closefd{}f'(\rho)a +\frac{1}{2}\tr b^*\closefd{}f'(\sigma)b.
\]
The map
\[
(\rho,h)\to\tr h^*\closefd{}f'(\rho)h
\]
is thus mid-point convex and by continuity therefore convex. This implies that $ f $ is a matrix entropy \cite[Theorem 2.1]{kn:hansen:2014:3}. 
\end{proof}

\begin{theorem}\label{main theorem}
Let $ f $ be an entropic function normalised such that $ f(1)=0, $ $ f'(1)=1 $ and $ f''(1)=1. $ Then $ f(t)=t\log t $ for $ t>0. $
\end{theorem}

\begin{proof}
By the preceding theorem it follows that $ f $ is a matrix entropy. The defining notion of a non-affine matrix entropy given by $ f $ is concavity in positive definite $ \rho $ of the map
\[
\rho\to \closefd{}f'(\rho)^{-1},
\]
cf. \cite[Definition 2.2]{kn:tropp:2014} and \cite[Definition 1.1]{kn:hansen:2014:3}. In particular, we obtain that the positive function
\[
g(t)=\frac{1}{f''(t)}\qquad t>0
\]
is concave.
Since $ f $ is also subentropic we know from Corollary~\ref{corollary: super-additive function} that $ g $ is super-additive with continuous extension to $ [0,\infty). $ Therefore,
\[
\frac{g(t+s)-g(t)}{s}\ge\frac{g(s)-g(0)}{s}\qquad t,s>0,
\]
and this inequality contradicts concavity of $ g $ for $ s<t $ unless $ g $  is affine. Since $ g(0)=0 $ there exists thus a constant $ b>0 $ such that
\[
f''(t)^{-1}=g(t)=bt\qquad t>0,
\]
and since $ f''(1)=1 $ we obtain that
\[
f''(t)=\frac{1}{t}\qquad t>0.
\]
Since $ f'(1)=1 $ we thus obtain $ f'(t)=\log t + 1, $ and since $ f(1)=0 $ finally
\[
f(t)=t\log t\qquad t>0
\]
which is the assertion.
\end{proof}

{\bf Acknowledgments.}  It is a pleasure to thank Bernhard Baumgartner and the anonymous referees for encouragement and for valuable suggestions.
The author also acknowledges support from the Japanese government Grant-in-Aid for scientific research 26400104.

{\small


}

\end{document}